\documentclass[11pt]{article}
\newcommand{\icalpremove}[1]{#1}
\newcommand{\icalpversion}[1]{}
\usepackage{amsmath, amsthm, amssymb,bbm,accents}
\usepackage{xspace,nicefrac}
\usepackage{fullpage}
\usepackage{xcolor,ifpdf}
\usepackage[compact]{titlesec}
\usepackage{authblk}

\definecolor{Darkblue}{rgb}{0,0,0.4}
\definecolor{Brown}{cmyk}{0,0.81,1.,0.60}
\definecolor{Purple}{cmyk}{0.45,0.86,0,0}

\ifpdf
 
\fi
\hypersetup{colorlinks=true,
            citebordercolor={.6 .6 .6},linkbordercolor={.6 .6 .6},%
citecolor=Darkblue,urlcolor=black,linkcolor=red,pagecolor=black}


\usepackage{setspace}
\newenvironment{subproof}[1][\proofname]{%
  \begin{proof}[#1]%
}{%
  \end{proof}%
}

\newcommand{\f}{\frac}
\newcommand{\nf}{\nicefrac}

\newcommand{\A}{{\mathcal{A}}}
\newcommand{\cD}{{\mathcal D}}
\newcommand{\BE}{\begin{enumerate}}
\newcommand{\EE}{\end{enumerate}}

\newcommand{\BI}{\begin{itemize}}
\newcommand{\EI}{\end{itemize}}

\newcommand{\Sum}{\displaystyle\sum\limits}

\newcommand{\R}{\mathbb R}
\newcommand{\Z}{\mathbb Z}

\newcommand{\I}{{\mathbb I}}
\newcommand{\calS}{{\mathcal S}}
\newcommand{\cO}{{\mathcal O}}

\newcommand{\B}{{\mathcal B}}
\newcommand{\eps}{\varepsilon}
\newcommand{\e}{\varepsilon}

\newcommand{\al}{\alpha}

\newcommand{\m}{\mathcal}

\newcommand{\lf}{\lfloor\!\!\lfloor}
\newcommand{\rf}{\rfloor\!\!\rfloor}

\newcommand{\poly}{\text{poly}}

\newcommand{\lmt}{\left[\begin{matrix}}
\newcommand{\rmt}{\end{matrix}\right]}

\newcommand{\imm}{{\tt immed}}
\newcommand{\jif}{\ensuremath{J^-_f}\xspace}
\newcommand{\jof}{\ensuremath{J^+_f}\xspace}
\renewcommand{\emptyset}{\varnothing}
\newcommand{\ts}{\textstyle}

\usepackage[colorinlistoftodos,textsize=tiny,textwidth=2cm,color=red!25!white,obeyFinal]{todonotes}

\newtheorem{theorem}{Theorem}
\numberwithin{theorem}{section}
\newtheorem{lemma}[theorem]{Lemma}
\newtheorem{claim}[theorem]{Claim}
\newtheorem{corollary}[theorem]{Claim}

\author[1]{\large Anupam Gupta}
 \author[2]{\large Amit Kumar}
 \author[1]{\large Jason Li}
 \affil[1]{Carnegie Mellon University} 
 \affil[2]{IIT Delhi} 

\title{Non-Preemptive Flow-Time Minimization via Rejections}
\begin{document}

\maketitle
\begin{abstract}
We consider the online problem of minimizing weighted flow-time on unrelated machines.  Although much is known about this problem in the resource-augmentation setting, these results assume that jobs can be preempted. We give the first constant-competitive algorithm for the non-preemptive setting in the rejection model. In this rejection model, we are allowed to reject an $\e$-fraction of the total weight of jobs, and compare the resulting flow-time to that of the offline optimum which is required to schedule all jobs. This is arguably the weakest assumption in which such a result is known for weighted flow-time on unrelated machines. While our algorithms are simple, we need a delicate dual-fitting argument to bound the flow-time. 
  \end{abstract}

\section{Introduction}
\label{sec:introduction}

Consider the problem of scheduling jobs  for weighted flow-time
minimization. Given a set of $m$ unrelated machines, jobs arrive online
and have to be processed on one of these machines. Each job $j$ is
released at some time $r_j$, has a potentially different processing
requirement (size) $p_{ij}$ on each machine $i$, and a weight $w_j$
which is a measure of its importance.  The objective function is the
\emph{weighted flow time} (or \emph{response time}): if the job $j$
completes its processing at time $C_j$, the flow/response time is
$(C_j - r_j)$, i.e., the time the job spends in the system. The goal is now to minimize the weighted sum
$\sum_j w_j(C_j - r_j)$. 

The problem of flow-time minimization has been extensively studied
both from theoretical and practical perspectives. 
The
theoretical analyses have to assume that the jobs can be pre-empted in
order to prove any meaningful competitive ratio, and it is easy to see
why. If we schedule a long low-weight job and a large number of short
high-weight items arrive meanwhile, we cannot afford to delay the latter (else we
suffer large flow-time), so the only solution would be to preempt the
former (See~\cite{KTW-journal} for strong lower bounds.) And even with
pre-emption, the problem turns out to be difficult for multiple
machines: e.g.,~\cite{GargK07-focs} show no bounded competitive ratio is
possible for the case of unrelated machines. Hence, it is natural to
consider models with ``resource augmentation'' where the algorithm has slightly
more resources than the adversary. E.g., in the speed-augmentation
setting, where  the algorithm uses machines of speed $(1+\e)$-times those 
of the adversary, Chadha et al.~\cite{CGKM09} showed how to get a
preemptive schedule with weighted flow time at most $\poly(1/\e)$ times
the optimal flow time.

A different model of resource augmentation was proposed by Choudhury et
al.~\cite{CDGK-jour} in the context of load balancing and maximum
weighted flow-time, where we are allowed to reject at most $\e$-fraction of
the total weight of the incoming jobs, but we compare with the optimum
off-line algorithm which is required to process all the jobs. The
motivation was two-fold: (a) the model is arguably more natural, since
it does not involve comparing to an imaginary optimal schedule running on
a slower machine, and (b) even with speed-augmentation, there are
problems, e.g. on-line load balancing, where even a constant factor
speed-up does not suffice to give meaningful results. Indeed, getting a
non-preemptive schedule for weighted flow-time is one of these problems.
Consider for example the following input: a job of unit size and unit
weight at time 0 arrives. As soon as the algorithm schedules it, the
adversary releases $L$ jobs of size $\e \ll 1/L^2.$ The optimal off-line
flow-time is $O(1)$, but the algorithm will incur total flow-time of
$\Omega(L)$.  The model of job rejection is intuitively more powerful
than speed-augmentation (although no such formal connection is known):
loosely, the speed-augmentation model only allows us to uniformly reject
an $\e$-fraction of each job, whereas the rejection model allows us to
``non-uniformly'' reject an arbitrary subset of jobs, as long as they
contribute only an $\e$-fraction of the total weight.

\subsection{Our Results}
We consider the problem of non-preemptive scheduling on unrelated machines
where the objective is to minimize total weighted flow-time of jobs. 
Our main result is the following:

\begin{theorem}[Main Theorem]
  \label{thm:main1}
  For the problem of online weighted flow-time minimization on unrelated
  machines, there is a deterministic algorithm that rejects at most an
  $\e$-fraction of the total weight of incoming jobs, and ensures that
  the total weighted flow time for the remaining jobs is at most an
  $O(1/\e^3)$ factor times the optimal weighted flow time without
  rejections.
\end{theorem}
Note that we compare with the off-line optimum which is allowed to be preemptive (in fact,
migratory), but is required to process all the jobs. 
Our guarantees are, in fact, stronger. Define the notion of a ``departure
time'' $D_j$ for the job, which is the time at which either the job
completes non-preemptively (in which case $D_j = C_j$) or is the time at
which the job is rejected. A different natural definition of the total
weighted response time in the presence of rejections would be the
following:
\[ \text{total weighted response time } := \sum_j w_j (D_j - r_j). \]
Keeping this quantity small forces us to decide on jobs early, and
discourages us from letting jobs linger in the system for a long time,
only to reject them at some late date. (Such a behaviour would be very
undesirable for a scheduling policy, and would even be considered
``unprofessional'' in real-world settings.)

\medskip\noindent
In fact the bulk of our work is in handling the single machine case.
For this case, we get a slightly stronger bound.
\begin{theorem}[Single Machine]
  \label{thm:main2}
  For the problem of online weighted flow-time minimization on
  a single machine, there is a deterministic algorithm that rejects at most an
  $\e$-fraction of the total weight of incoming jobs, and ensures that
  the total weighted flow time for the remaining jobs is at most 
  $O(1/\e^2)$ factor times the optimal weighted flow time without
  rejections even when the offline optimum  is given $(1+\e)$-extra speedup. 
\end{theorem}
The
fact that we can compare with an optimum offline algorithm which has
faster machine allows us to use known immediate-dispatch algorithms for
the setting of unrelated machines in a black-box manner~\cite{CGKM09, AGK12}.

\subsection{Our Techniques}
\label{sec:techniques}

Let us first focus on the single-machine case. Our algorithm rejects
jobs in two different ways: some of the jobs are rejected immediately
upon arrival, and others are rejected after receiving some processing.
Moreover, assume for the moment that we are running a preemptive
schedule, but without speed-augmentation. The high-level idea is to
reject a ``random'' $\e$-fraction of jobs that come in. At an intuitive
level, this rejects only $\e$-fraction of the weight (although this only
in expectation, whereas we want this to hold deterministically at all
times), and should create the effect of $\e$-speed augmentation.
To implement this, let $\alpha_j$ be the ``effect'' of job $j$ on the
system---i.e., the increase in the total flow-time of the jobs currently
in the system (assuming no future jobs arrive).  The value of $\alpha_j$
also naturally corresponds to settings of dual variables for a natural
flow-time LP. Using this we can (more-or-less) show that (a)~the
$\alpha_j$ values of the rejected jobs give us a lower bound on OPT,
whereas (b)~the $\alpha_j$ values of the non-rejected jobs upper-bound
our cost. Hence, our goal becomes: at each time cancel \emph{at most} an
$\e$-fraction of the total weight $\sum_j w_j$, while cancelling
\emph{at least} an $\e$-fraction (say) of the total ``dual'' value
$\sum_j \alpha_j$.

A little thought shows that this abstract task is hopeless in general
for any deterministic strategy (say, if the $\alpha$ values rise very
sharply), so we have to take the structure of the $\alpha_j$ values into
account. We do this in two steps: we break the $\alpha_j$ contribution
into $\al_j^+$, the effect of job $j$ on items denser than $j$, and
$\al_j^-$, its effect on less-dense items. Now we put jobs into buckets
based on having the same $(\alpha^+, w)$ or $(\alpha^-, w)$ values, and
rejecting each $1/\e^{th}$ job in each bucket. (The actual bucketing is
a little finer, see \S\ref{sec:single-machine}.) Moreover, we reject
the first job in each $(\alpha^+, w)$ bucket. The complications arise
because we are more aggresive for each such $(\alpha^+, w)$ bucket, and
because we may not have rejected any jobs in the $(\alpha^-, w)$ if it
had less than $1/\e$ items. In \S\ref{sec:lemma-alphaminus} we perform a
delicate charging to relate our aggressive rejections for the former to
the total running time of the jobs, and show that (i)~this aggressive
rejection does not reject too much weight, and (b)~also compensates for
our timid rejections in the latter bucketing.

This high-level argument was done assuming preemptions. Since we want a
non-preemptive schedule, only immediate rejections do not suffice, and
we also must reject some jobs which we have started processing---indeed,
if a large number of high-density (``important'') jobs arrive right
after we start processing some long low-density job $j$, delaying these
more important jobs would cause large flow-time. So we must reject job
$j$. However, as long as the total weight of these new jobs is $w_j/\e$,
we can charge the rejection to these new jobs. This rejection
makes the schedule very ``unstable'' and hence complicates the analysis.
To get around this problem, we \emph{mark} the job $j$ as
``preemptible''. We then run a version of HDF with some preemptible and
other non-preemptible jobs, and show that its performance can also be
related to the LP variables.

Finally, for the multiple machines case we can perform a modular
reduction to the single-machines case. We first use the \emph{immediate
  dispatch} algorithm of Anand et al.~\cite{AGK12} to assign jobs to
machines, assuming speed augmenation. We then show our algorithm does
well even compared to a stronger benchmark (i.e., where the offline
schedule---instead of the online schedule---gets the speed
augmentation). This gives us the theorem for the unrelated machines.


\subsection{Related Work}
There has been considerable work on the problem of minimizing total flow-time in the
online setting, though most of it is in the preemptive setting. Several logarithmic competitive
 algorithms are known for unweighted flow-time on identical machines setting~\cite{LR, AvrahamiA03},
and in the related machines setting~\cite{GK06-icalp, Anand13}, 
but there are strong lower bounds for the case of weighted flow-time even on
a single machine~\cite{BC09}. In the restricted assignment settings with preemption, the unweighted flow-time problem
becomes considerably harder even for 3 machines~\cite{CGKM09}. The situation for non-preemptive flow-time is much 
harder. Kellerer et al.~\cite{KTW-journal} showed that one cannot achieve $o(n)$-competitive algorithm even for a single machine. 

Much stronger results are known in the speed augmentation model, where machines in the online algorithm 
have $\e$-fraction more speed than the corresponding machines in the offline setting. This model was first 
proposed by Kalyanasundaram and Pruhs~\cite{KP00} for the problem of non-clairvoyant preemptive total flow-time 
minimization on  a single machine. They gave an $O(1/\e)$-competitive algorithm for this problem. 
Chadha et al.~\cite{CGKM09} gave $O(1/\e^2)$-competitive preemptive algorithm for weighted flow-time in the unrelated 
machines setting. This was extended to the non-clairvoyant setting by Im et al.~\cite{ImKMP14}. However, the non-preemptive
weighted flow-time problem has strong lower bounds in the speed augmentation model even on a single machine~\cite{LucarelliTST16}. 

The rejection model was proposed by Choudhury et al.~\cite{CDGK-jour} in the context of load balancing and
maximum weighted flow-time in the restricted assignment setting.  Lucarelli et al.~\cite{LucarelliTST16} 
considered the non-preemptive scheduling problem 
of minimizing  weighted flow-time in the unrelated machines setting. They showed that one can get $O(1/\e)$-competitive
algorithm if we allow both $(1+\e)$-speed augmentation and rejection of jobs of total weight $\e$-times the total weight. 
Assuming both, we can design a much simpler algorithm and use the dual fitting techniques developed for speed augmentation models
to give a simple analysis of this algorithm (see the comment after Lemma~\ref{lem:opt-lower}). 
Independently of us, 
Lucarelli et al.~\cite{LucarelliTST17} recently announced an algorithm where they can 
remove the speed augmentation assumption for the simpler unweighted setting. 

In the {\em prize-collection} model,   one is allowed to incur a penalty term 
for the rejected jobs. This model has been widely studied, see e.g. Bartal et al.~\cite{BartalLMSS00}, Eppstein et al.~\cite{Epstein14}, and
Bansal et al.~\cite{Bansal03}, though is considerably different from our model because here one can reject a large
fraction of the jobs.

\section{Definitions and Preliminaries}

We consider the unrelated machine scheduling problem, as defined in
\S\ref{sec:introduction}. Our schedules will be non-preemptive. For a schedule
$\calS$, let $C^\calS_j$ denote the completion time of $j$. We use $F^\calS_j$ to
denote the flow-time of $j$, and the objective function is given by
$F^\calS:= \sum_j w_j \cdot F^\calS_j$. We may remove the superscript $\calS$ if
it is clear from the context. We use $\cO$ to denote the optimal
off-line schedule.  In Section~\ref{sec:single-machine}, when
considering the special case of a single machine, we use $p_j$ to denote
the \emph{processing time} of job $j$ (on this machine). Define the
\emph{density} $\rho_j$ of a job as the ratio $w_j/p_j$.
We assume that the parameter $\eps$ satisfies $\e^2 \leq 1/2$, and that
$\nf1\e \in \Z$.

\medskip\noindent\textbf{Fractional weighted flow-time.} Given a schedule $\A$, let $p_j(t)$ denote the remaining processing time of 
job $j$ at time $t$ (assuming $t \geq r_j$).  The remaining weight of $j$ at time $t$ is defined as $w_j(t):=\rho_j \cdot p_j(t)$. 
The weighted flow-time of $j$ in this schedule is defined as $w_j(C_j - r_j)$, where $C_j$ is the completion time of $j$. The fractional weighted flow-time
of $j$ is defined as $\sum_{t \geq r_j} w_j(t). $ Since $w_j(t) = 0$ for
$t \not\in [r_j, C_j]$, and $w_j(t) \leq w_j$ for any time $t$, it is clear
that the fractional weighted  flow-time is at most the (integral) weighted flow-time of $j$. The following claim is easy to check. 
\begin{claim}
\label{cl:wtdef}
If a job $j$ is processed without interruption during $[t, t+p_j]$, then
its fractional weighted flow-time is $w_j (t-r_j) + \nf{w_j
  p_j}2$. Moreover, if a job $j$ gets rejected at time $t'$, its
weighted fractional flow-time is at least $\nf{w_j(t'-r_j)}2$.
\end{claim}
Since the integral weighted flow-time of a job as in the claim above is
$w_j (t-r_j)+w_jp_j$, we see the integer and fractional flow times are
within factor of 2 of each other. Thus, for jobs which do not get
preempted, we can argue about weighted fractional flow-time.

\section{Algorithm for Single-Machine Weighted Flow Time}
\label{sec:single-machine}

In this section, we consider the single-machine setting.
For ease of algorithm description, we assume that all quantities are integers so that we can 
schedule jobs at the level of integer time-slots. We first describe an algorithm $\A$ which both rejects and preempts jobs. 
We subsequently show how to modify this algorithm (in an online manner)
to another schedule which only rejects jobs, and does no preemptions.  
During our algorithm, we shall say that a job $j$ is {\em active} at time $t$ if it has
been released by time $t$, but has not finished processing until time
$t$, and has not been rejected. Let $A(t)$ denote the set of active jobs at time $t$ in our algorithm. 
A subset of these jobs, denoted by $L(t)$, will be special---these jobs
are allowed to be preempted (at time $t$). Once a job enters the set $L(t)$ at some 
time $t$, it stays in $L(t')$ for all subsequent times $t' \geq t$ until
it finishes processing.

For a job $j \in A(t)$ and time $t$, recall that $p_j(t)$  denotes the
remaining processing time.
At every point of (integer) time $t$, the algorithm performs the following steps (in this order):
\begin{enumerate}
\item If job $j$ arrives at time $t$, the algorithm may choose to reject
  it immediately upon arrival. We will call such rejections {\em
    immediate rejections.}  If the job is not rejected, it gets added to
  the active set $A(t)$. For the moment, this is the only way in which a
  job gets
  rejected. 
\item Let $j$ be the job getting processed just before time $t$ (i.e.,
  in the time-slot $[t-1,t]$). If job $j$ was not already in the set
  $L(t)$, the algorithm may move it to the set $L(t)$ if ``many'' jobs
  smaller than $j$ have arrived during its execution. We will specify
  the precise rule soon. Recall that once added, the job $j$ will remain
  in the set $L(t)$ until it finishes.
\item If the job $j$ getting processed in the time-slot $[t-1,t]$ did
  not finish at time $t$ and it is not in $L(t)$, the algorithm will
  continue to process $j$ during the next time-slot $[t,
  t+1]$. Otherwise, if $j$ finishes or $j \in L(t)$, the algorithm 
  chooses a job in $A(t)$ which has the highest density (the HDF rule)
  and processes it during $[t, t+1]$.
\end{enumerate}

Note that if multiple jobs arrive at a time $t$, we consider them in
arbitrary order, and carry out the first two steps above iteratively for
each such job, before executing step~3.
This completes the description
of the algorithm, except that we have not specified the rules for the
first two steps. 

We first explain the rule for adding a job to $L(t)$.  
Suppose the algorithm processes a job $j$ during $[t-1, t]$, and suppose $j \notin L(t-1)$. 
Let $t'$ be the time when the algorithm started processing $j$. Since it was not allowed to preempt $j$, 
it must have processed $j$ without interruption during $[t',t]$. 
If the total weight of jobs
arriving during $(t',t]$ exceeds $w_j/\eps$, we add job $j$ to the set
$L(t)$.
The intuition behind this rule is simple---the final algorithm will
eventually reject all jobs which get added to the set $L(t)$, for all
$t$.  We can charge the weight of the rejected job $j$ to the weight of
the jobs which arrived during $[t',t]$. Moreover, consider a job $j$
that does not get added to $L(t)$ over its lifetime. In a preemptive
setting, we may have preempted such a job $j$ on the arrival of a new
shorter job, whereas here we perform such a preemption only when enough
shorter jobs arrive. Since $j$ was not added to $L(t)$, the total weight
of such shorter jobs waiting on $j$ is at most $w_j/\eps$, so we can pay
for the additional flow-time incurred by these shorter jobs (up to an
$1/\eps$ factor) by the flow-time of $j$.

The rule for immediate rejections is more involved. We maintain two tables $T^+$ and
$T^-$. 
Each arriving job may get assigned to either $T^+$ or $T^-$, or
both.  We refer to each entry of these tables as a {\em bucket}. At a
high level, every $(1/\eps)^{th}$ job arriving in each bucket in either
table suffers immediate rejection, though the details differ for the two
tables. Let us elaborate on this further.

With every newly arriving job $j$, we specify a quantity $\alpha_j$, which is the
increase in the total flow-time of all the jobs in the system, assuming
(i)~no further jobs arrive after job $j$, and (ii)~the scheduling
algorithm follows the preemptive HDF policy from $r_j$ onwards for {\em
  all} the jobs in $A(r_j)$.
As in~\cite{AGK12}, we can write an expression 
for $\alpha_j$ as follows. 
\begin{gather}
  \al_j:= \Big( w_j\Sum_{j'\in A(r_j):\; \rho_{j'} \geq \rho_j }
  p_{j'}(r_j) \Big) ~~+~~ \nf{w_jp_j}{2} ~~+~~ \Big( p_j \Sum_{j'\in A(r_j):\;
  \rho_{j'} < \rho_j } w_{j'}(r_j) \Big). \label{eq:alpha}
\end{gather}
We establish the convention that $A(r_j)$ does not contain job
$j$. Moreover, if multiple jobs are released at time
$r_j$, we consider them in arbitrary but fixed order, and add only those
jobs to $A(r_j)$ which are considered before $j$.

For $x \in \R$, let $\lf x \rf$ denote the largest integer
$i$ such that $2^i \leq x$. For a job $j$, define its {\em density-class}
as $\lf \rho_j \rf$. We partition jobs in $A(r_j)$ depending on their
density-class as follows:
\begin{gather}
  D_j^+ := \{ j' \in A(r_j) \mid \lf \rho_{j'} \rf \geq \lf \rho_j \rf
  \} \qquad\text{and}\qquad D_j^- := \{j' \in A(r_j) \mid \lf \rho_{j'}
  \rf < \lf \rho_j \rf \}. \label{eq:8}
\end{gather}
Now let
$\al_j^+$ be the terms in the expression for
$\al_j$ involving jobs in $D_j^+$, and define
$\al_j^-$ similarly. In other words,
\begin{align}
  \al_j^+ &:=\Big( w_j\sum_{j'\in D^+_j\,:\; \rho_{j'} \geq \rho_j
  }p_{j'}(r_j) \Big) + \Big( p_j \sum_{j'\in D_j^+\,:\; \rho_{j'}
  < \rho_j}w_{j'}(r_j) \Big); \label{eq:9} \\
  \al_j^- &:= \Big( p_j \sum_{j'\in D^-_j\,:\; \rho_{j'}
  < \rho_j }w_{j'}(r_j) \Big). \label{eq:10}
\end{align}
Clearly, $\al_j=\al_j^++\nf{w_jp_j}2 +\al_j^-$.
We now specify the definitions of the two tables. 

\begin{itemize}
\item Table $T^+$: Buckets in this table are indexed by ordered pairs of
  integers $(\kappa, \lambda)$. If an
  arriving job $j$ satisfies $\al_j^+ \geq w_j p_j/\eps$, we assign it
  to the bucket indexed $(\lf \nf{\al_j^+}{w_j}\rf,\lf w_j\rf)$ in this
  table, and add it to the set $J^+$ of jobs assigned to $T^+$.  For
  each bucket, we cancel the first job that is assigned to that bucket,
  and then every $(1/\e)^{th}$ subsequent job assigned to it.
\item Table $T^-$: Buckets in this table are indexed by ordered triplets
  of integers $(\gamma, \delta, \eta)$.  Each arriving job which
  satisfies $\al_j^- > w_jp_j/\eps$ is assigned to the bucket indexed
  $(\lf \al_j^- \rf,\lf \rho_j \rf,\lf p_j \rf)$, and added to the set
  $J^-$ of jobs assigned to $T^-$. For each bucket, cancel every
  $(1/\e)^{th}$ job assigned to this bucket. Note the subtle difference
  with respect to $T^+$: here the first job to be canceled in a bucket
  is the $(1/\e)^{th}$ job assigned to it.
\end{itemize}

\subsection{The Final Algorithm $\B$} The actual online algorithm $\B$
is almost the same as $\A$, except when the algorithm $\A$ processes a
job in $L(t)$ during time-slot $[t, t+1]$, the algorithm $\B$ idles,
leaving this slot empty. In other words, when a job being executed is
added to $L(t)$, the algorithm $\B$ rejects the job instead of
eventually finishing it, perhaps after some preemptions. (We can think
of this as being a \emph{delayed rejection}, as opposed to the
\emph{immediate rejection} that $\A$ performs based on the above
bucketing strategy.) Clearly, we can implement $\B$ in an online manner.

\section{Analyzing the Single-Machine Algorithm}

In this section, we provide the analysis of our single-machine algorithm $\m B$. Naturally, the two main steps are to show that (i) an $O(\e)$ fraction of jobs by weight get rejected, and (ii) the total flow time is competitive with the optimal offline algorithm.

Showing (i) is relatively straightforward: a rejected job is either immediately rejected or is later rejected in $\m B$ due to its preemption in $\m A$. We will show that the rejected jobs falling under each of the two categories is an $O(\e)$ fraction by weight, with a separate analysis for each category. Both of the analyses are in Section~\ref{sec:bound-weight-single-mc}.

To show flow time competitiveness of algorithm $\m B$, we instead focus on bounding the tota flow time of algorithm $\m A$.
 By Claim~\ref{cl:wtdef}, the total (integer)
flow-time of jobs that $\B$ does not reject is within a factor of two of their
fractional flow-time in $\A$, since these are precisely the jobs that
$\A$ does not preempt. Therefore, to prove Theorem~\ref{thm:main2}, it suffices to show that $\m A$ is $O(1/\e^2)$ factor competitive with the optimal offline algorithm.

Let $J^\imm$ denote the set
of jobs which get rejected immediately upon arrival, and let $\cO$ denote the optimal
offline schedule and $F^\cO$ its fractional weighted flow time. Roughly speaking, our goal is to establish the following chain of approximate inequalities:
\begin{gather} \e F^{\m A} \lesssim \e\sum_j\al_j \lesssim 
\sum_{j\in J^\imm}\al_j \lesssim F^\cO , \label{eq:chainIneq}\end{gather}
where $\lesssim$ hides additive $\sum_j\nicefrac{w_jp_j}{\e^{O(1)}}$ factors. Since $F^\cO \geq \sum_j w_jp_j/2$, these additive losses still provide a $\nicefrac1{\e^{O(1)}}$ competitive ratio.

For the first inequality, we will bound the flow time of algorithm $\m A$, modulo an additive $\sum_j\nicefrac{w_jp_j}{\e}$ factor, by the sum of $\al_j$ over all jobs $j\notin J^\imm$, which are precisely the jobs that are finished by $\m A$. We do so by exploiting the facts that the $\al_j$ values indicate an increase in flow time to an HDF algorithm, and that $\m A$ is ``approximately'' an HDF algorithm. The details are in Lemma~\ref{lem:algA}.

The second inequality is the most technically involved section of the paper. Not only does the immediate rejection scheme reject an $O(\e)$ fraction of jobs, but it also rejects jobs constituting an $\e$ fraction of the total $\al_j$ value. The analysis is in Section~\ref{sec:alpha-terms}.

Finally, the last inequality relates the optimal offline flow time to the sum of the $\al_j$ values of immediately rejected jobs. It is restated as Lemma~\ref{lem:opt-lower} and proved in the appendix.

\subsection{Bounding Weight of Rejected Jobs}
\label{sec:bound-weight-single-mc}

In this section, we show that the total weight of rejected jobs is only an $O(\e)$ fraction of total. Recall that
jobs either suffer immediate rejection, or are added to $L(t)$ for some
time $t$, and hence suffer delayed rejection.

Let us first bound the total weight of the set $L:= \cup_t L(t)$. For a
job $j$ in $L(t)$, let $s_j$ be the first time when it gets processed
and $l_j$ be the time at which it enters the set $L(t)$. Since $j$ must
be processed uninterrupted in this interval $(s_j, l_j]$, the intervals
associated with different jobs are disjoint. Moreover job $j$ entered
$L(t)$ because the total weight of jobs released during $(s_j, l_j]$ is
at least $w_j/\e$. Thus the total weight of jobs in $L$ can be upper
bounded by $\e$ times the weight of all the jobs. 

We now account for the weight of jobs which are rejected immediately on
arrival. For job $j$, let $\lf w_j \rf$ denote the {\em weight-class} of
this job. Jobs assigned to a bucket in $T^+$ have the same weight-class,
by construction of the buckets. Jobs assigned to a bucket in $T^-$ have
the same $\lf \rho_j \rf$ and $\lf p_j \rf$, which pins down their
weight $w_j = \rho_j \cdot p_j$ up to a factor of $4$. This gives us the
following facts:

\begin{itemize}
\item Since we reject every $(1/\e)^{th}$ job in each bucket of $T^-$,
  the total weight of jobs in $J^-$ which get rejected immediately is at
  most $4\e$ times the weight of all jobs in $J^-$.
\item Let $\jof$ be the subset of jobs in $J^+$ which happen to be the
  first jobs to be assigned to their respective buckets in $T^+$. Then
  the weight of all jobs in $J^+ \setminus \jof$ which get rejected
  immediately on arrival is at most $2\e$ times the total weight of all
  the jobs in $J^+$.
\end{itemize}

So it remains to account for the items items in $\jof$, which are all
rejected. Recall that a job in $J^+$ is assigned to the bucket indexed
$(\lf \al_j^+/w_j \rf, \lf w_j \rf)$ in $T^+$.  Jobs in $\jof$ are
assigned to distinct buckets in $T^+$. Fix an integer $\gamma$, and let
$J^\gamma$ denote the jobs in $\jof$ which are mapped to a bucket
indexed $(\gamma, \kappa)$ for some $\kappa$. The jobs in $J^\gamma$
have distinct weight-classes and so it suffices to bound the weight of
the highest weight job in $J^\gamma$---let this heaviest job be
$j^\gamma$. Let $S$ denote the set of such jobs $j^\gamma$ as we range
over all $\gamma$. Jobs in $S$ have distinct $\lf \al_j^+/w_j \rf$
values. Let $\Gamma= \{\gamma_1 < \gamma_2 < \ldots < \gamma_k\}$ be the
integers $\gamma$ for which there is a job $j^\gamma \in S$, and let the
corresponding jobs in $S$ be called $j_1, j_2, \ldots, j_k$.

Now starting from the smallest index in $\Gamma$, we charge each job
$j_r \in S$ to a subset of jobs of total weight at least $w_{j_r}/\e$.  The
job $j_r$ may charge to a job fractionally---if it charges to a fraction
$\delta$ of some job $j$, then it can only use $\delta w_j$ amount of
weight of $j$ for its charging (and we say that ``\emph{$j_r$ charges to
  $\delta p_j$ size of this job $j$}''). Of course, we need to ensure
that the total fraction charged to a job is at most 1. We inductively
maintain the following invariant for all $r \in 1\ldots k$:
\begin{itemize}
\item The job $j_r$ charges to jobs of total (fractional)
  weight at least $\nf{w_{j_r}}{8\e}$.
\item Jobs $j_1, \ldots, j_r$ charge to jobs of total (fractional) size
  at most $2^{\gamma_r}$.
\end{itemize}
Assuming these invariants hold for $r-1$, we show that they hold for $r$
as well. Let $\rho^\star := \lf \rho_{j_r} \rf$ be the density class for job
$j_r$. By $j_{r}$'s choice of bucket, $\lf \alpha_{j_{r}}/w_{j_r} \rf =
\gamma_r$, so 
\begin{gather}
  \al_{j_{r}}^+ \geq 2^{\gamma_{r}} \cdot w_{j_r}. \label{eq:11}
\end{gather}
Recall from~(\ref{eq:8}) that $D_{j_{r}}^+$ is the set of jobs of
density class $\rho^\star$ or higher which are active at the time $j_{r}$
is released. Let $P_r := \sum_{j \in D_{j_{r}}^+} p_j$ be the total
processing time of these jobs. By (\ref{eq:9}), it follows that
\begin{gather}
  \al_{j_r}^+ \leq w_{j_r} P_r. \label{eq:12}
\end{gather}
Combining~(\ref{eq:11}) and~(\ref{eq:12}), $P_r \geq 2^{\gamma_{r}}$.
By the second invariant, the first $r$ jobs $j_1, \ldots, j_{r-1}$ have only
charged to jobs of total size at most $2^{\gamma_{r-1}}$, so we can find
jobs in $D_{j_r}^+$ of total (fractional) size
$2^{\gamma_{r}}-2^{\gamma_{r-1}} \geq 2^{\gamma_{r}-1} $ which have not
been charged yet, and charge to them. This proves the second invariant.

To prove the first invariant, we know that $\al_{j_r}^+ \geq w_{j_r}
p_{j_r}/\e$, else $j_r$ would not be assigned to $T^+$. Moreover,
$\al_{j_r}^+ \leq w_{j_r} 2^{\gamma_{r}+1}$ by the bucketing, so
$2^{\gamma_{r}} \geq p_{j_r}/2\e$. Consequently, we charge to jobs of
total size at least $2^{\gamma_{r}-1} \geq p_{j_r}/4\e$, and these jobs
have density class at least $\rho^\star$. Since $2\rho^\star \geq
\rho_{j_r} = w_{j_r}/p_{j_r}$, we get their total (fractional) weight is at
least $w_{j_r}/8\e$. This proves the first invariant, and hence the
following theorem.

\begin{theorem}[Few Rejections]
  \label{thm:wt-rej}
  The weight of jobs suffering immediate rejection, plus those in
  $\cup_t L(t)$, is at most an $O(\e)$ fraction of the weight of all
  jobs released.
\end{theorem}

\subsection{Bounding the Weighted Fractional Flow-time}
\label{sec:bound-flow-time}

Next we show that the total fractional flow-time of $\A$ can be bounded
in terms of total $\alpha_j$ values.
We first focus on relating $F^{\m A}$ to the sum of the $\al_j$ values, as described in (\ref{eq:chainIneq}).

Observe that $\alpha_j$ denotes the increase in objective function due
to the arrival of $j$ if we had followed the preemptive HDF policy for
all the jobs from time $r_j$ onwards. However, we follow a slightly
different policy---if $j'$ denotes the job that was running on the
machine at time $j$'s release time $r_j$, we let $j'$ run until it
finishes, or else until $j'$ belongs to the set $L(t')$ at some time
$t' \geq r_j$. If no further jobs are released after $j$, the HDF policy
after this time $t'$ would be non-preemptive.  Thus, we would still
expect that the total fractional weighted flow-time of our algorithm to
be close to $\sum_j \alpha_j$. We formalise this intuition now. For
every job $j$, we define a job $\phi(j)$ as follows: let $j'$ be the job
which was running just before time $r_j$ (i.e., in the slot
$[r_j-1, r_j]$). If $j' \notin L(r_j)$, we define $\phi(j)$ to be $j'$,
otherwise we leave $\phi(j)$ undefined.  Our policy for adding a job to
the set $L(t)$ ensures that for every job $j$, $w(\phi^{-1}(j))$ is at
most $w_j/\eps$.~\footnote{For a set $S$ of jobs, let $w(S)$ denote the
  total weight of jobs in $S$.} Recall that $J^\imm$ is the set
of jobs which get rejected immediately upon arrival. The following lemma
states that the fractional weighted flow-time of the algorithm
can be charged to the $\alpha_j$ values of the jobs which get immediately rejected. 

\begin{lemma}
  \label{lem:algA}
The fractional weighted flow-time of $\A$  is at most $\sum_{j: j \notin J^\imm} \alpha_j + \sum_j w_j p_j /\eps. $
\end{lemma}

\begin{proof}
  Jobs in $J^\imm$ get rejected immediately, so their flow-time is 0. We
  now consider the jobs which are not immediately rejected in the rest
  of the proof.  Consider the jobs in order of increasing release
  times. Let $\Delta_j$ denote the increase in the objective function
  value due to arrival of $j$.
  In other words, if $J_1$ is the set of
  jobs released before $j$, then $\Delta_j$ equals the total fractional
  weighted flow-time of $\A$ on the input $J_2 := J_1 \cup \{j\}$ minus
  that on the input $J_1$. The total weighted flow time of $\A$ on the
  entire input would be $\sum_j \Delta_j$, the sum of these
  increases. We now show that
  \begin{gather}
    \Delta_j \leq \alpha_j + w_j p_{\phi(j)}. \label{eq:1}
  \end{gather}
  Since $w(\phi^{-1}(j')) \leq w_{j'}/\eps$, we get that $\sum_j
  w_j p_{\phi(j)} = \sum_{j'} w(\phi^{-1}(j')) p_{j'} \leq \sum_{j'}
  w_{j'} p_{j'}/\e$. Hence, summing~(\ref{eq:1}) over all $j$ which are not in 
  $J^\imm$ proves the lemma. 

  Now we prove~(\ref{eq:1}). Since we will be dealing with two inputs,
  $J_1$ and $J_2$, we parameterise all quantities by $J_1$ or $J_2$ to
  clarify which input we refer to. For example, $A(J_k, t), k=1, 2$ will
  refer to the active set $A(t)$ on input $J_k$. 
  Let
  $F(J_k, t)$ denote the fractional weighted flow-time of jobs in
  $A(J_k, t)$ {\em beyond} time $t$, i.e.,
  $F(J_k, t) := \sum_{t' \geq t} \sum_{j \in A(J_k, t')} w_j(t')$.

  There are two cases when job $j$ arrives. If $\phi(j)$ is undefined,
  the job $j'$ running in slot $[r_j-1,r_j]$ belongs to $L(r_j)$. Hence
  the algorithm $\A$ on both inputs $J_1, J_2$ just runs HDF starting at
  time $r_j$. The difference between the corresponding flow times is
  precisely $\alpha_j$, by definition.

  Otherwise $\phi(j)$ is well-defined.  Since $j$ is the latest arrival,
  the job $\phi(j)$ will not be preempted, and runs to completion. Say
  job $\phi(j)$ completes at time $t'$. During the time $[r_j, t']$ the
  difference in fractional weighted flow-time between the two runs is
  precisely $w_j\cdot(t' - r_j) \leq w_j p_{\phi(j)}$. After time $t'$
  we run HDF on the remaining jobs, and the difference in the fractional
  weighted flow-time of the two runs is precisely what $\alpha_j$ would
  have been had $j$ arrived at time $t'$ instead of time $r_j$. In other
  words, if $J' := A(J_k, r_j) \setminus \{\phi(j)\}$,
  \begin{align*}
    F(J_2, t') - F(J_1, t') &= \nf{w_j p_j}2 + \sum_{j' \in J':
    \rho_{j'} \geq \rho_j} w_{j} p_{j'}(t') + \sum_{j' \in J':
                              \rho_{j'} < \rho_j} w_{j'}(t') p_{j} \\
                            &= \nf{w_j p_j}2 + \sum_{j' \in J':
    \rho_{j'} \geq \rho_j} w_{j} p_{j'}(r_j) + \sum_{j' \in J':
                              \rho_{j'} < \rho_j} w_{j'}(r_j) p_{j} 
  \end{align*}
  But this is a subset of the terms of $\alpha_j$: indeed, we're just
  missing the term corresponding to job $\phi(j)$. Hence, the total
  difference is at most $\alpha_j + w_jp_{\phi(j)}$, proving~(\ref{eq:1}).
\end{proof}

To bound our flow time against the optimum using this lemma, note that
$\sum_j w_j p_j/\e \leq 2F^\cO/\e$, where we recall that $\cO$ denotes the optimal
offline schedule, and $F^\cO$ its fractional weighted flow time. So we
just need to bound
$\sum_j \alpha_j= \sum_j \nf{w_j p_j}2 + \sum_j \al_j^+ + \sum_j \al_j^-$. The
first term is again bounded by $F^\cO$, so the work is in bounding the other two
terms. We first record a convenient lemma -- its proof is based on LP duality arguments
and construction of dual variables are similar to those in~\cite{AGK12}.

\begin{lemma}[Duality-based Lower Bound on OPT]
  \label{lem:opt-lower}
  $\sum_{j \in J^\imm} \alpha_j   \leq F^\cO + \sum_j \nf{w_j p_j}{\e}$.
\end{lemma}

\begin{proof}
        Consider the linear program for fractional weighted flow time (note that the variables $x_{j,t}$ are
        only defined for $t \geq r_j$):
  \begin{alignat*}{2}
    \min  \ts \sum_{t, j} w_j \left( \frac{t - r_j}{p_j} + \frac12 \right) & \, x_{t,j} &\qquad\qquad & \\
    \ts \sum_t \nf{x_{t,j}}{p_j} &\geq 1  & & \forall \text{ jobs } j\\
    \ts\sum_j x_{t,j} &\leq 1 & & \forall \text{ times } t \\
    x_{t,j} &\geq 0. && \forall j,t
  \end{alignat*}
  The dual is
  \begin{alignat*}{2}
    \max \ts \sum_j \alpha_j &- \ts \sum_t \beta_t  & \qquad\qquad &\\
    \ts \f{\alpha_j}{p_j} - \beta_t &\leq \ts \frac{w_j (t - r_j)}{p_j}
    +  \frac{w_j}2
    &&\forall j,t \\
    \alpha_j, \beta_t &\geq 0. &&
  \end{alignat*}
  Weak duality implies that any feasible dual solution value is at most
  the optimal primal solution value, which in turn is at most
  $F^\cO$. Define $\alpha_j$ as above, and let
  $\beta_t := \sum_{j \in A(t)} w_j(t)$ be the total fractional weight
  in the system at time $t$. Therefore $\sum_t \beta_t$ is the total
  weighted fractional flow-time of $\A$. Lemma~\ref{lem:algA} upper
  bounds
  $\sum_t \beta_t \leq \sum_{j \not\in J^\imm} \alpha_j + \sum_j
  \nf{w_jp_j}{\e}$, and the dual objective function is at least
  $\sum_{j} \alpha_j - \sum_t \beta_t \geq \sum_{j \in J^\imm} \alpha_j
  - \sum_j \nf{w_jp_j}{\e}$. The desired result will follow once we
  prove that the dual variables are feasible.
  
  To show feasibility, consider job $j$ released at time $r_j$, and a
  time $t \geq r_j$.  Let $\beta_t'$ denote the total remaining weight
  of jobs at time $t'$ if no jobs arrive after $j$ and we run preemptive
  HDF from time $r_j$ onwards. (Recall that $\alpha_j$ captures the
  increase in fractional weighted flow-time due to arrival of $j$
  precisely in this scenario). We show that the dual constraint for
  the pair $j,t$ is satisfied with $\beta_t$ replaced by
  $\beta_t'$. This suffices because HDF has the property that at any
  time during the schedule, the residual weight of jobs is minimized
  compared to any other algorithm, and hence $\beta_t \geq \beta_t'$.
   
  Now we consider running HDF on $A(r_j)$ (excluding $j$) from time $r_j$ onwards. HDF orders these
  jobs according to density---let this ordering be $\prec$. 
  Suppose HDF
  processes a job $j'$ at time $t$. Two cases arise: (i) $j'$ appears before $j$ in the order $\prec$, or (ii) it appears
  after $j$ in this ordering. Consider case~(i) first. By splitting $j'$ into two parts (each of which has the same density
  as that of $j'$), we can assume that HDF starts processing  $j'$ at time $t$. Therefore, $t-r_j = \sum_{j'' \prec j'} p_{j''}(r_j)$. 
  Therefore, 
  \begin{align*}
  \frac{w_j (t-r_j)}{p_j} + \beta_t' & =  \rho_j \cdot \sum_{j'': j'' \prec j'} p_{j''}(r_j) + \sum_{j'': j'\preceq j''} w_{j''}(r_j) 
  & \geq \rho_j \cdot \sum_{j': j'' \prec j} p_{j''}(r_j) + \sum_{j'': j \prec j''} w_{j''}(r_j),
  \end{align*}
  where we have used the fact that if $j''$ satisfies
  $j' \preceq j'' \prec j$, then $\rho_{j''} \geq \rho_j$ and so,
  $w_{j''}(r_j) \geq \rho_j \cdot p_{j''}(r_j)$. The RHS above is
  precisely $\frac{\al_j}{p_j} - \frac{w_j}{2}$, which is what is wanted
  to prove. 
  
  For case~(ii), again assume that the algorithm just started processing $j'$ at time $t$. As above,
  \begin{align*}
  \frac{w_j (t-r_j)}{p_j} + \beta_t' & =  \rho_j \cdot \sum_{j'': j'' \prec j'} p_{j''}(r_j) + \sum_{j'': j'\preceq j''} w_{j''}(r_j) 
  & \geq \rho(j) \cdot \sum_{j': j'' \prec j} p_{j''}(r_j) + \sum_{j'': j \prec j''} w_{j''}(r_j),
  \end{align*}
     where we use the fact that if $j''$ satisfies $j \prec j''$, then $\rho_j \cdot p_{j''}(r_j) \geq \rho_{j''} \cdot p_{j''}(r_j) = w_{j''}(r_j). $
     As before, the RHS is precisely $\frac{\al_j}{p_j} -
     \frac{w_j}{2}$. This proves dual feasibility, and hence the lemma. 
   \end{proof}

If we were to also assume $(1+\e)$-speed augmentation, we can strengthen
the lower bound on $F^\cO \geq (1/\e) \sum_j \alpha_j$. Combined with
Lemma~\ref{lem:algA}, this immediately shows that the algorithm is
constant competitive---we do not even need any immediate
rejections to get this result.

\subsection{Controlling the $\alpha$ Terms}
\label{sec:alpha-terms}

In this section, our goal is to establish the approximate inequality $\e\sum_j\al_j \lesssim 
\sum_{j\in J^\imm}\al_j$, introduced in (\ref{eq:chainIneq})%
and made precise in Corollary~\ref{cor:alphas}.

\begin{lemma}
  \label{lem:plus}
  $\sum_j \alpha_j^+ \leq O(1/\eps) \cdot \big(\sum_j w_jp_j + \sum_{j \in
    J^\imm} \alpha_j^+ \big)$.
\end{lemma}

\begin{proof}
  The definition of $J^+$ implies that
  $\sum_{j \notin J^+} \al_j^+ \leq \sum_{j \notin J^+} w_j p_j/\e$.
  It remains to bound $\sum_{j \in J^+} \al_j^+$.  We do an accounting
  per bucket in $T^+$. Fix a bucket $B$ indexed by a pair
  $(\kappa, \lambda)$, i.e., all jobs $j$ in this bucket have
  $\lf \nf{\al_j^+}{w_j}\rf = \kappa, $ and $\lf w_j\rf =
  \lambda$. Hence, if $j$ is any job in this bucket, then
  $2^{\kappa} \leq \nf{\al_{j}^+}{w_{j}} \leq 2^{\kappa+1}, $ and
  $2^{\lambda} \leq w_j \leq 2^{\lambda+1}$. Multiplying,
  $2^{\kappa + \lambda} \leq \al_j^+ \leq 4 \cdot 2^{\kappa + \lambda}$,
  i.e., the $\al_j^+$ values of any two jobs in this bucket differ by a
  factor of at most 4.

  Let $J_B$ denote the jobs in $J^+$ assigned to this bucket $B$, and $n_B$
  denote their cardinality $|J_B|$.  Since we reject the first job and
  then every subsequent $(1/\eps)^{th}$ job in $J_B$, we immediately
  reject at least $\eps\, n_B$ jobs in $J_B$.  Therefore,
  $$ \sum_{j \in J_B} \al_j^+ \leq \frac{4}{\eps} \cdot  \sum_{j \in J_B
    \cap J^\imm} \al_j^+. $$ Summing over all buckets, the lemma follows. 
\end{proof}

\begin{lemma}
  \label{lem:minus}
  $\sum_j \alpha_j^- \leq O(1/\eps) \cdot
  \big(\sum_j w_jp_j + \sum_{j \in
    J^\imm} \alpha_j \big)$. 
\end{lemma}

\begin{proof}
  The argument is similar to Lemma~\ref{lem:plus} in spirit, but
  technically more involved. The reason is that we do not remove any
  jobs from a bucket of $T^-$ until it has $1/\eps$ jobs assigned to
  it. Hence, for a bucket $B$, if $J_B$ is non-empty but
  $|J_B| \leq 1/\e$, we have $J_B \cap J^{\imm} = \emptyset$. However,
  if \jif is the set of jobs in $J^-$ which are the first jobs assigned to
  their corresponding buckets in $T^-$, then we get (as in the proof of Lemma~\ref{lem:plus})
  that
  \begin{gather}
    \sum_j \al_j^- \leq O(1/\eps) \cdot \Bigg( \sum_j w_jp_j + \sum_{j
        \in J^\imm} \al_j^- + \sum_{j \in \jif} \al_j^-\Bigg). \label{eq:13}
  \end{gather}
  It remains to bound $\sum_{j \in \jif} \al_j^-$, which we accomplish via the following claim. Since the proof is more technical, we defer it to the next section.
  \begin{claim}
    \label{clm:minus-parttwo}
    $\sum_{j \in \jif} \al_j^- \leq O(\e) \cdot \big( \sum_j w_jp_j +
    \sum_j \alpha_j^+ \big)$.
  \end{claim}
  Combining this with~(\ref{eq:13}) and  Lemma~\ref{lem:plus}, using that
  $\al_j^+ + \nf{w_jp_j}2 + \al_j^- = \al_j$, the lemma follows.
\end{proof}

Using Lemmas~\ref{lem:plus} and~\ref{lem:minus}, we obtain the desired relation between $\sum_j\al_j $ and $\sum_{j\in J^\imm}\al_j$.
\begin{corollary}
  \label{cor:alphas}
  $\sum_j \alpha_j \leq O(1/\e) \cdot (\sum_{j \in J^\imm} \alpha_j + \sum_j w_jp_j)$. 
\end{corollary}

Finally, we put together the bounds on $\al_j$, establishing the chain of inequalities as described in (\ref{eq:chainIneq}) and bounding the competitive ratio of algorithm $\m A$.

\begin{theorem}
  \label{thm:flow-single}
  The fractional weighted flow-time of the non-rejected jobs in $\A$ is $O(F^\cO/\e^2)$.
\end{theorem}

\begin{proof}
  By Lemma~\ref{lem:algA}, the fractional weighted flow-time of the
  non-rejected jobs in $\A$ is at most $\sum_j (\alpha_j +
  w_jp_j/\e)$. This  is bounded by
  $O(1/\e) \cdot (\sum_{j \in J^\imm} \al_j + \sum_j w_jp_j)$ by
  Corollary~\ref{cor:alphas}. Finally, Lemma~\ref{lem:opt-lower}
  bounds this by $O(1/\e) \cdot (F^\cO + \sum_j w_jp_j/\e)$. Since $F^\cO \geq
  \sum_j w_jp_j/2$, this completes the proof.
\end{proof}

\subsubsection{Proof of Claim~\ref{clm:minus-parttwo}}
\label{sec:lemma-alphaminus}

In this section, we prove Claim~\ref{clm:minus-parttwo}, bounding the
$\alpha_j^-$ value of $\jif$. For brevity, 
define
$\Lambda^+ := \sum_j \al_j^+$.  Recall that for a job $j$, its density
class is given by $\lf \rho_j \rf = \lf \nf{w_j}{p_j} \rf$. For
each density class $\delta \in \Z$, let us define some notation:
\begin{itemize}
\item Let $A^\delta(t) := \{ j \in A(t) \mid \lf \rho_j \rf = \delta\}$
  denote jobs in $A(t)$ whose density class is $\delta$.  
\item Let $P^\delta(t) := \sum_{j \in A^\delta(t)} p_j(t)$ and
  $W^\delta(t) := \sum_{j \in A^\delta(t)} w_j(t)$ be the total
  processing time and residual weight of jobs in $A^\delta(t)$,
  respectively. Since all jobs in this set have the same density class,
  observe that $\frac{W^\delta(t)}{P^\delta(t)}$ also lies in the range
  $[2^\delta, 2^{\delta+1})$.  
\item Define $P^\delta := \max_t P^\delta(t)$
  and $W^\delta := \max_t W^\delta(t)$.
\end{itemize}
Our proof shows that $\sum_{\delta} P^\delta W^\delta$ is small; then we bound $\sum_{j \in \jif} \al^-_j$ by
$\sum_{\delta} P^\delta W^\delta$.

\begin{lemma}
  \label{lem:densityclass}
   $\sum_{\delta} P^\delta
  W^\delta \leq O(1) \cdot \left( \sum_j w_j p_j + \Lambda^+ \right). $ 
\end{lemma}

\begin{proof}
  Let us first prove an analogous statement for any fixed time $t$,
  which we can then extend to prove the desired statement.

\begin{claim}
  \label{clm:densityclass}
  For any time $t$ and density class $\delta$,
  $ P^\delta(t) W^\delta(t) \leq O(1) \cdot \sum_{j \in A^\delta(t)} (w_j p_j + 
    \al_j^+ ).$
\end{claim}

\begin{subproof}
  To this end, arrange the jobs in $A^\delta(t)$ in \emph{decreasing} order
  $j_1, \ldots, j_k$ of their arrival time.  
  At $r_{j_\ell}$, the arrival time of $j_{\ell}$, all jobs in
  $J_{\ell+1} := \{j_{\ell+1}, \ldots, j_k\}$ are in $A(r_{j_\ell})$ but
  $j_\ell$ is not. 
  Consider an arbitrary job $j' \in J_{\ell+1}$. The contribution of
  $j'$ towards $\al_{j_{\ell}}^+$ is at least the minimum of
  $p_{j'}(r_{j_\ell}) w_{j_\ell}$ and $p_{j_\ell}
  w_{j'}(r_{j_\ell})$. Since both $j'$ and $j_\ell$ have the same
  density class, this is at least $\nf{p_{j_\ell}
    w_{j'}(r_{j_\ell})}2$. The residual weight is non-increasing over
  time and $r_{j_\ell} \leq t$, so this is at least
  $\nf{p_{j_\ell} w_{j'}(t)}2$.

  Summing over all $j' \in J_{\ell+1}$ (and adding in $w_{j_\ell} p_{j_\ell}$)
  \begin{gather}
    \alpha_{j_\ell}^+ + w_{j_\ell} p_{j_\ell} \geq p_{j_\ell} 
      \sum_{j \in J_\ell} w_j(t)/2 \geq 2^{\delta-1} \; w_{j_\ell}(t) \sum_{j
        \in J_\ell} w_j(t).
  \end{gather}
  Summing over $\ell=1, \ldots, k$
  $$\sum_{\ell = 1}^k \left( \alpha_{j_\ell}^+ + w_{j_\ell} p_{j_\ell}
  \right) \geq 2^{\delta-1} \, \sum_{\ell=1}^k w_{j_\ell}(t) \sum_{i 
        = \ell}^k w_{j_i}(t) \geq  
  2^{\delta-1} \cdot \frac{W^\delta(t)^2}4 \geq \f{P^\delta(t) W^\delta(t)}{16}. $$
  The second inequality above uses the fact that if $n_1, \ldots, n_k$ are positive reals, then 
  \begin{gather*}
  \sum_{\ell=1}^k n_\ell \cdot (n_{\ell} + \ldots + n_k) \geq 1/4 \cdot
    (n_1 + \ldots + n_k)^2. \qedhere
  \end{gather*}
\end{subproof}
Let $t$ and $t'$ be such that $P^\delta = P^\delta(t)$ and
$W^\delta = W^\delta(t')$. Since all jobs in
$A^\delta(t') \cup A^\delta(t)$ have densities within factor of 2 of
each other,
$W^\delta(t) \geq 2^\delta\, P^\delta(t) \geq 2^\delta\, P^\delta(t')
\geq W^\delta/2$. The result now follows from
Claim~\ref{clm:densityclass}, and observing that $w_jp_j$ and
$\alpha_j^-$ are both non-negative for the remaining jobs.
\end{proof}

\begin{lemma}
  \label{lem:alminus}
  $\sum_{j \in \jif} \alpha_j^- \leq O(\e) \cdot \sum_\delta P^\delta W^\delta$.
\end{lemma}

\begin{proof}
  Let us first give a general method for bounding $\al_j^-$ of any job
  $j \in J^-$, and then we can apply it to the jobs in
  $\jif \subseteq J^-$. Recall that the jobs which contribute to
  $\al_j^-$ are the ones with a strictly smaller density class than that
  of $j$. We now show that one need not look at jobs of all such
  classes, and a subset of these classes suffice. Fix a job $j \in J^-$ of
  density class $\delta$, and define an index set $\I_j$ as 
 follows:
  \begin{gather}
    \I_j := \{ \theta < \delta \mid P^\theta(r_j) \geq (1.5)^{\delta -
      \theta} p_j /8\e \}.
  \end{gather}

  \begin{claim}
    \label{clm:alminus}
    For any job $j \in J^-$ with density class $\delta$, $\al_j^- \leq 4 p_j \cdot \sum_{\theta \in \I_j} W^\theta. $
  \end{claim}
  \begin{subproof}
    Let $j'$ be a job in $A(r_j)$ of strictly lower density class than
    $j$. Its contribution towards $\al_j^-$ is $p_j w_{j'}(r_j) $. Therefore,
    $\al_j^- $ is at most
    \begin{gather}
      \sum_{\theta < \delta} p_j W^\theta(r_j) = p_j \cdot \sum_{\theta
        \in \I_j} W^\theta(r_j) + p_j \cdot \sum_{\theta \notin \I_j,
        \theta < \delta} W^\theta(r_j). \label{eq:5}
    \end{gather}
    Let us bound the summation from the second expression.
    \begin{gather}
      \sum_{\theta \notin \I_j, \theta < \delta} W^\theta(r_j) \leq
      \sum_{\theta \notin \I_j, \theta < \delta} 2^{\theta + 1} \,
      P^\theta(r_j) \leq \sum_{\theta < \delta}
      \frac{(1.5)^{\delta-\theta}}{2^{\delta-\theta}} \cdot
      \frac{2^\delta p_j}{4\e} \leq\frac{3 w_j}{4 \e}. \label{eq:4}
    \end{gather}
    Substituting~(\ref{eq:4}) into~(\ref{eq:5}), and using that
    $\al_j^- \geq w_jp_j/\e$ for all jobs $j \in J^-$, we get that 
    $\nf{\alpha_j}4 \leq p_j \sum_{\theta \in \I_j} W^\theta(r_j) \leq p_j \sum_{\theta \in \I_j} W^\theta$, which  proves the desired result.
\end{subproof}
Recall that job $j \in J^-$ is mapped in table $T^-$ to the bucket
indexed by $(\lf \al_j^- \rf, \lf \rho_j \rf, \lf p_j \rf).$ For a fixed
pair $(\delta, \eta)$, consider the jobs in $\jif$ which are mapped to
buckets indexed $(\gamma, \delta, \eta)$ with various values of
$\gamma$, and denote these jobs by $J_{(\delta, \eta)}$. Since $\jif$
only contains the first job in each bucket, the $\lf \al_j^- \rf $
values of the various jobs in $J_{(\delta, \eta)}$ are all distinct. It
follows that if $j^\star$ is the job in $J_{(\delta, \eta)}$ with the
highest $\al_j^-$ value, then
$\sum_{j \in J_{(\delta, \eta)}} \al_j^- \leq 4\al_{j^\star}^-$. Thus,
we just need to worry about one job per $J_{(\delta, \eta)}$---let $S$
denote this set of jobs.

The ordered pairs $(\lf \rho_j \rf, \lf p_j \rf)$ corresponding to jobs
$j \in S$ are all distinct. For density class $\delta$, let $S^\delta$
denote the jobs in $S$ with density class $\delta$. Using
Claim~\ref{clm:alminus},
\begin{gather}
  \sum_{j \in S^\delta} \al_j^- \leq 4 \sum_{j \in S^\delta} p_j
  \sum_{\theta \in \I_j} W^\theta = 4 \sum_{\theta < \delta} W^\theta
  \sum_{j \in S^\delta : \theta \in \I_j} p_j. \label{eq:6}
\end{gather}
The jobs in $S^\delta$ also have different $\lf p_j \rf$ values, so the
sum $\sum_{j \in S^\delta : \theta \in \I_j} p_j \leq 4p_{j'}$ for the
job $j' := \arg\max\{ p_j \mid j \in S^\delta, \theta \in \I_j\}$. By
definition of $\I_j$, $p_{j'} \leq 8\e P^\theta/(1.5)^{\delta-\theta}$.
Substituting into~(\ref{eq:6}), 
\begin{gather}
  \sum_{j \in S^\delta} \al_j^- \leq 16 \sum_{\theta < \delta} \frac{8 \eps
    \; W^\theta P^\theta}{(1.5)^{\delta-\theta}}. \label{eq:7}
\end{gather}
To complete the argument, 
\begin{align*}
  \sum_{j \in \jif} \alpha_j^- \leq 4 \sum_{\delta} \sum_{j \in
  S^\delta} \alpha^-_j &\stackrel{\text{eq.(\ref{eq:7})}}{\leq}  
   2^9 \e \sum_\delta \sum_{\theta < \delta} \frac{W^\theta
  P^\theta}{(1.5)^{\delta-\theta}} \\
  &= 2^9 \e \sum_{\theta} W^\theta P^\theta \cdot \sum_{\delta > \theta} \frac{1}{(1.5)^{\delta-\theta}}
    = O\Big( \eps \sum_{\theta} W^\theta P^\theta\Big).
\end{align*}

This completes the proof of Lemma~\ref{lem:alminus}.
\end{proof}
Combining Lemmas~\ref{lem:densityclass} and~\ref{lem:alminus} completes
the proof of Claim~\ref{clm:minus-parttwo}, and hence for
Theorem~\ref{thm:flow-single}. 
In Section~\ref{sec:augmentation}, we show that the algorithm is
competitive even against an optimal algorithm that is allowed
$(1+\e)$-speed augmentation---and hence prove Theorem~\ref{thm:main2}.
\section{Comparing with off-line optimum with speed augmentation}
\label{sec:augmentation}

We now consider the case when the optimal algorithm is allowed $(1+\e')$ speed augmentation;
here $\e'$ will be $O(\e)$, and show that our algorithm is competitive
even with this stronger benchmark.  Let $\cO^{\e'}$ denote the new
optimal solution. Our algorithms $\A$ and $\B$ remain
unchanged. Lemma~\ref{lem:algA} remains unchanged because the definition
of $\al_j$ is the same. Lemma~\ref{lem:opt-lower} now gets modified as
follows. 
\begin{lemma}
  \label{lem:opt-lower-2}
  $\sum_{j \in J^\imm} \alpha_j - 2\e' \cdot \sum_j \alpha_j - 2\sum_j
  \nf{w_j p_j}{\e} \leq F^{\cO^{\e'}}$.
\end{lemma}
\begin{proof}
  The LP relaxation for the off-line optimum with $(1+\e')$-speed
  augmentation is same as that in Lemma~\ref{lem:opt-lower} except that
  the constraint for each time $t$ changes to
  $$\ts\sum_j x_{t,j} \leq 1+\e'  \qquad\qquad \forall \text{ times } t $$
  As a result, the constraints in the dual objective function remain
  unchanged, but the dual objective value changes to
  $\sum_j \al_j - (1+\e') \cdot \sum_t \beta_t$.  Our definitions of
  $\al_j, \beta_t$ remain unchanged, and so, dual feasibility still
  holds. Since $\sum_t \beta_t$ denotes the total fractional weighted
  flow-time of the jobs, Lemma~\ref{lem:algA} shows that the dual
  objective value is at least
  $\sum_{j \in J^\imm} \al_j - 2\e' \sum_j \al_j - 2\sum_j \nf{w_j
    p_j}{\e}. $
\end{proof}

We are now ready to state the main result comparing against this
stronger benchmark. 
\begin{theorem}
  \label{thm:augment}
  The total fractional weighted flow-time, and hence the total weighted
  flow-time of non-rejected jobs, is
  $O(F^{\cO^{\e'}}\!\!/\e^2)$.
\end{theorem}
\begin{proof}
  Corollary~\ref{cor:alphas} and Lemma~\ref{lem:opt-lower-2} imply that
  \[ F^{\cO^{\e'}} \geq \Omega\Big(\e \cdot \sum_j \al_j - \sum_j
    \nf{w_j p_j}{\e} \Big), \] and so
  $\sum_j \al_j \leq O \big( F^{\cO^{\e'}}\!\!/\e + \sum_j \nf{w_j
    p_j}{\e^2})\big)$.  Lemma~\ref{lem:algA} implies that the fractional
  flow time of our algorithm is greater by at most $\sum_j
  w_jp_j/\e$. Since $\sum_j w_jp_j/(2(1+\e)) \leq F^{\cO^{\e'}}$,
  and the fractional and integral weighted flow-time are within factor
  of 2 of each other if we consider jobs which are not preempted, we get
  the theorem.
\end{proof}

\section{Extension to Unrelated Machines}

The extension of our result on single machine to the more general
scenario of unrelated machines can be done very modularly. 
\icalpremove{Recall that
in the unrelated machines setting, there are $m$ machines, and job $j$
has processing requirement $p_{ij}$ on machine $i$. For a subset of jobs
$S$ and parameter $\e' > 0$, let $\cO^{\e'}\!\!(S,i)$ denote the optimal
off-line solution to jobs in $J$ when we only consider machine $i$
(i.e., jobs in $J$ have processing time $p_{ij}$ on this single
machine), and we also augment this machine to have speed $(1+\e')$. Let
$F^{\cO^{\e'}\!\!(S,i)}$ denote the total weighted flow-time of this
solution. Let $J$ denote the entire input set of jobs. }
We shall use the
following result from~\cite{CGKM09,AGK12}.
\begin{theorem}
  \label{thm:flow}
  There is an online algorithm $\cD$ which dispatches each arriving job
  $j$ immediately upon arrival to one of the $m$ machines such that the
  following property holds: if $J^{(i)}$ is the set of jobs which are
  dispatched to machine $i$,
  then $\sum_i F^{\cO^{\e'}\!\!(J^{(i)},i)}$ is the optimal solution 
  to $J^{(i)}$ when we have only one machine with speed $(1+\e')$, 
  at most $\nf{1}{\e'}$ times the optimal weighted flow-time of $J$.
\end{theorem}

 The algorithms in~\cite{CGKM09,AGK12} actually build a schedule as
well and use this schedule to immediately dispatch a job. The algorithm
$\cD$ can build this schedule in the {\em background} and use it to
dispatch jobs, but not use it for actual processing. It follows from
Theorem~\ref{thm:flow} and Theorem~\ref{thm:augment}
that if we run our
algorithm on each of the machines $i$ (with input $J^{(i)}$ arriving
on-line) independently, then the total weighted flow-time of
non-rejected jobs in our algorithm is at most $O(1/\e^3)$ times the
optimal value. This proves Theorem~\ref{thm:main1}. 

\section{Conclusion}

We have given the first algorithm for  minimizing weighted  flow-time in the non-preemptive setting in the rejection model. It remains an interesting open problem to extend this
result to (weighted) $\ell_p$ norms of flow-time for values of $p > 1$, and in particular, for non-preemptive weighted maximum flow-time. 

\bibliographystyle{plain}
\bibliography{sched}

\end{document}